\documentclass[10pt, a4paper]{article}

\usepackage{xcolor}
\definecolor{MyBlue}{cmyk}{1,0.13,0,0.63}
\definecolor{MyGreen}{cmyk}{0.91,0,0.88,0.52}
\newcommand{\mylinkcolor}{MyBlue}
\newcommand{\mycitecolor}{MyGreen}
\newcommand{\myurlcolor}{webbrown}

\usepackage{hyperref}
\hypersetup{%
  bookmarksnumbered=true,bookmarksopen=false,%
  plainpages=false,
  linktocpage=true,%
  colorlinks=true,breaklinks=true,%
  linkcolor=\mylinkcolor,citecolor=\mycitecolor,urlcolor=\myurlcolor,%
  pdfpagelayout=OneColumn,%
  pageanchor=true,%
}

\usepackage{graphicx}
\usepackage{amsmath, amsthm,  amsfonts, amscd}
\usepackage{amssymb}
\usepackage{url}
\usepackage{fullpage}
\usepackage{mathtools}
\usepackage[all]{xy}
\usepackage{slashed}

\newtheorem{thm}{Theorem}[section]
\newtheorem*{thm*}{Theorem}

\newtheorem{lemma}[thm]{Lemma}
\newtheorem{prop}[thm]{Proposition}

\theoremstyle{definition}
\newtheorem{defn}[thm]{Definition}
\theoremstyle{remark}
\newtheorem{remark}[thm]{Remark}

\numberwithin{equation}{section}

\newcommand{\End}{\ensuremath{\mathrm{End}}}

\newcommand{\wh}{\ensuremath{\widehat}}

\newcommand{\N}{\ensuremath{\mathbb{N}}}
\newcommand{\Z}{\ensuremath{\mathbb{Z}}}

\newcommand{\C}{\ensuremath{\mathbb{C}}}

\def\calT{\mathcal{T}}

\def\calK{\mathcal{K}}

\def\calH{\mathcal{H}}

\def\calA{\mathcal{A}}

\def\calS{\mathcal{S}}
\def\calE{\mathcal{E}}

\def\calZ{\mathcal{Z}}

\newcommand{\ol}{\overline}

\theoremstyle{definition}

\DeclareMathOperator{\Dom}{Dom}

\DeclareMathOperator*{\res}{res}

\newcommand{\rst}[1]{\ensuremath{{\mathbin\upharpoonright}%
\raise-.5ex\hbox{$#1$}}}

\makeatletter

\newcommand{\Rmnum}[1]{\expandafter\@slowromancap\romannumeral #1@}
\makeatother

\makeatletter
\renewcommand\@biblabel[1]{#1.}
\makeatother

\usepackage[T1]{fontenc}
\usepackage[utf8]{inputenc}
\usepackage{authblk}

\title{The bulk-edge correspondence for the quantum Hall effect in Kasparov theory}
\author[$\dagger \ddagger$]{Chris Bourne\thanks{chris.bourne@anu.edu.au, alan.carey@anu.edu.au, renniea@uow.edu.au}}
\author[$\dagger\ddagger$]{Alan L. Carey}
\author[$\ddagger$]{Adam Rennie}
\affil[$\dagger$]{Mathematical Sciences Institute, Australian National University, Canberra, ACT 0200, Australia}
\affil[$\ddagger$]{School of Mathematics and Applied Statistics, University of Wollongong, Wollongong, NSW 2522, Australia}

\date{\today}

\begin{document}

\maketitle

\begin{abstract}
We prove 
the bulk-edge correspondence in $K$-theory for the quantum Hall effect by 
constructing an unbounded Kasparov module from a short exact sequence 
that links the bulk and boundary algebras. This approach allows us to 
represent bulk topological invariants explicitly as a Kasparov 
product of boundary invariants with the extension class linking the algebras. 
This paper focuses on the example of the discrete integer quantum Hall effect, 
though our general method potentially has much 
wider applications.
\end{abstract} 
\hspace{-0.1cm} \\
Keywords: quantum Hall effect, spectral triples, 
$KK$-theory, bulk-edge correspondence. \\ 
Mathematics subject classification: Primary 81V70, Secondary 19K35.

\section{Introduction}
In this letter, we revisit the notion of the bulk-edge correspondence 
in the discrete (or tight binding) version of the integer quantum Hall 
effect as previously studied in~\cite{EG02, EGS05, SBKR00, SBKR02, KSB04a, KSB04b}. 
In these papers, the motivation is to incorporate the presence of a 
boundary or edge into Bellissard's initial explanation of the quantum 
Hall effect~\cite{Bellissard94}. This is done by introducing 
an `edge conductance', $\sigma_e$, which is then shown to 
be the same as Bellissard's initial expression for the (quantised) Hall 
conductance, $\sigma_H$. Our motivation comes from the more 
$K$-theoretic arguments used in~\cite{SBKR02, KSB04b}. 

We propose a new method based on explicit representations of 
extension classes as Kasparov modules. Given a short exact sequence of $C^*$-algebras,
$$  0 \to  J \to A \to A/J \to 0  $$
for some closed $2$-sided ideal $J$, we know by results of 
Kasparov~\cite{Kasparov80} that this gives rise to a class in 
$\mathrm{Ext}(A/J,J)$, which is the same as $KK^1(A/J, J)$ 
for the algebras we study. By representing our short exact 
sequence as an unbounded Kasparov module, we can use 
the methods developed in~\cite{BMvS13, KL13, MeslandMonster} 
to take the Kasparov product of our module with spectral triples 
representing elements in $K^j(J)\cong KK^j(J,\C)$ to give elements 
in $K^{(j +1)}(A/J,\C)$. 

In this letter we focus on a simple case so as not to obscure the main
idea with technical details.  Thus we consider
 the short exact sequence representing the Toeplitz extension 
of the rotation algebra, $\calA_\phi$. An unbounded Kasparov module can be built from this extension by considering the circle action on the rotation algebra $\calA_\phi$, as in~\cite{CNNR}. 

We outline an alternative method for constructing a Kasparov module representing an extension class (generalised in~\cite{RRS15}) via a singular functional. We introduce this method with a view towards more complicated examples, where the circle-action picture breaks down. Such examples include the following.
\begin{enumerate}
  \item For the case of a finite group $G$ with $K\triangleleft G$, the short exact sequence 
  $$  0 \to J \rtimes K \to A \rtimes G \to A/J \rtimes G/K \to 0  $$
  can no longer be represented by circle actions. Such crossed products may emerge by considering the symmetry group of topological insulator systems, for example.
  \item For models with internal degrees of freedom (such as a honeycomb lattice), we would no longer be working with the Cuntz-Pimsner algebra of a self Morita-equivalence bimodule (as defined in~\cite[Section 2]{RRS15}) and so the singular functional method is necessary.
\end{enumerate}

See~\cite{RRS15} for more examples of extensions requiring this viewpoint. 
The flexibility of our approach to representing extensions as 
Kasparov modules (with which products can be taken) will allow many more systems-with-edge
to be
investigated, as we outline below. 

\subsection{Statement of the main result}
We begin with a Toeplitz-like extension of the rotation algebra $\calA_\phi$, 
and show how to  construct an unbounded Kasparov module 
$\beta=\left(\calA_\phi, Z_{C^*(\wh{U})},N\right)$ that represents 
this extension in $KK$-theory. Here $Z_{C^*(\wh{U})}$ is a Hilbert 
$C^*$-module coming from the extension, $\wh{U}$ is the shift 
operator on $\ell^2(\mathbb Z)$ along the boundary $\mathbb Z$ 
and the unbounded operator $N$ is a number operator (defined later). 

We also introduce a `boundary spectral triple' $\Delta=\left(C^*(\wh{U}),\ell^2(\Z),M\right)$, 
which we think of as the standard spectral triple over the circle but in a 
Fourier transformed picture (so that $M$ is the Fourier transform of differentiation and $\wh{U}$ is
the bilateral shift). 
Our main result, 
Theorem \ref{thm:unitary_equiv_between_product_module_and_bulk_spec_trip}, is as follows.
\begin{thm*}
The internal Kasparov product $\beta\hat\otimes_{C^*(\wh{U})}\Delta$ is 
unitarily equivalent to the negative of the spectral triple modelling the boundary-free quantum Hall effect.
\end{thm*}
We note that the Kasparov product and unitary equivalence of the Kasparov modules considered in the theorem is at the unbounded level, a stronger equivalence than in the bounded setting.

Recall from the work of Bellissard~\cite{Bellissard94} that the 
quantised Hall conductance in the case without boundary comes 
from the pairing of the Fermi projection with an element in $K^0(\calA_\phi)$. 
Our main result says that this $K$-homology class can be `factorised' into a 
product of a $K$-homology class representing the boundary and a $KK^1$-class 
representing the short exact sequence linking the boundary and boundary-free systems. 
We can then use the associativity of the Kasparov product to immediately obtain an 
edge conductance, and the equality of the bulk and edge conductances. 

It is in this point that our work differs from, but complements, the boundary 
picture developed in~\cite{SBKR02,KSB04b}, where the authors had to define a 
separate edge conductance and then show equality with the usual 
Hall conductance. Instead, our method derives the bulk-edge 
correspondence as a direct consequence of the factorisation of the 
boundary-free $K$-homology class. Our work demonstrates how we can obtain 
the bulk-edge correspondence of~\cite{SBKR02} without passing to cyclic homology and cohomology. This 
allows our method to be applied to systems with torsion invariants, which cannot be detected in cyclic theory. This is essential for topological insulator theory, where torsion invariants arise
naturally.

We also note that by working in the unbounded $KK$ picture, all computations are explicit. 
As Kasparov theory can also be extended to accommodate group actions and
real/Real algebras this means our method has potential applications to a 
much wider array of physical models. Topological insulators are an
example of  where the bulk-edge correspondence needs further work.

The paper is organised into two major Sections.  Section 2 contains the 
construction of the Kasparov module that is needed in Section 3 where the 
main theorem is proved.  Some details are relegated to an Appendix.

\subsubsection*{Acknowledgements}
All authors acknowledge the support of the Australian Research Council. 
The authors are grateful to the Hausdorff Institute for Mathematics for 
support to participate in the trimester program on `Noncommutative geometry 
and its applications' where some of this work was completed. 
AC thanks the Erwin Schr\"odinger Institute for support.
CB thanks Friedrich-Alexander Universit\"{a}t for hospitality during
the writing of this letter. All authors thank Koen van den Dungen and 
Hermann Schulz-Baldes for useful discussions.


\section{A Kasparov module representing the Toeplitz extension}

\subsection{The setup and the Pimsner-Voiculescu short exact sequence}
Recall~\cite{CM96} that in the discrete or `tight binding' model of the quantum Hall effect without boundary, we have magnetic translations $\wh{U}$ and $\wh{V}$ as unitary operators on $\ell^2(\Z^2)$. These operators commute with the unitaries $U$ and $V$ that generate the Hamiltonian $H=U+U^* + V + V^*$. We choose the Landau gauge such that 
\begin{align*}
   &(\wh{U}\lambda)(m,n) = \lambda(m-1,n),  &&(\wh{V}\lambda)(m,n)=e^{-2\pi i\phi m}\lambda(m,n-1), \\
   &(U\lambda)(m,n) = e^{-2\pi i \phi n}\lambda(m-1,n),  &&(V\lambda)(m,n) = \lambda(m,n-1),
\end{align*}
where $\phi$ has the interpretation as the magnetic flux through a unit cell and $\lambda\in\ell^2(\Z^2)$. We are keeping the model simple in order to make our constructions as clear as possible, though what we do extends to more sophisticated models. We note that $\wh{U}\wh{V}=e^{2\pi i\phi}\wh{V}\wh{U}$ and $UV = e^{-2\pi i\phi}VU$, so $C^*(\wh{U},\wh{V})\cong \calA_\phi$, (the irrational rotation algebra when $\phi$ is irrational), and $C^*(U,V)\cong \calA_{-\phi}$. We can also interpret $\calA_{-\phi} \cong \calA_{\phi}^{\mathrm{op}}$, where $A^\mathrm{op}$ is the opposite algebra with multiplication $(ab)^\mathrm{op} = b^\mathrm{op} a^\mathrm{op}$. Our choice of gauge also means that $C^*(\wh{U},\wh{V})\cong C^*(\wh{U})\rtimes_\eta \Z$, where $\wh{V}$ is implementing the crossed-product structure via the automorphism $\eta(\wh{U}^m) = \wh{V}^* \wh{U}^m \wh{V}$.

We outline an idea loosely based on that of Kellendonk et al.~\cite{SBKR02, KSB04b}, who employed constructions from Pimsner and Voiculescu~\cite{PV80}. The essence of the idea is to relate the bulk and edge algebras via a Toeplitz-like extension. This viewpoint is also  employed in~\cite{MT15b}.

\begin{prop}[\S2 of~\cite{PV80}] \label{prop:PV_sequence_Circle_case}
Let $S$ be the usual shift operator on $\ell^2(\N)$ with $S^*S =1$, $SS^* =1-P_{n=0}$. There is a short exact sequence,
$$  
  0 \to C^*(\wh{U})\otimes \calK[\ell^2(\N)] \xrightarrow{\psi} C^*(\wh{U}\otimes 1,\wh{V}\otimes S) \to C^*(\wh{U})\rtimes_\eta \Z \to 0. 
$$
\end{prop}
The map $\psi$ given in Proposition \ref{prop:PV_sequence_Circle_case} is such that
$$  \psi(\wh{U}^m \otimes e_{jk}) = (\wh{V}^*)^j \wh{U}^m \wh{V}^k \otimes S^j P_{n=0} (S^*)^k  $$
for matrix units $e_{jk}$ in $\calK[\ell^2(\N)]$. It is then extended to the full algebra by linearity. One checks that $\psi$ is an injective map into the ideal of $C^*(\wh{U}\otimes 1, \wh{V}\otimes S)$ generated by $1\otimes P_{n=0}$. We also have the isomorphism $C(S^1)\rtimes_\eta \Z \cong C^*(\wh{U}\otimes 1, \wh{V}\otimes V)\cong C^*(\wh{U},\wh{V})$, where $V$ is the image of $S$ under the map to the Calkin algebra. These alternate but equivalent presentations of $\calA_\phi$ will be of use to us later. For convenience, we denote $\calT = C^*(\wh{U}\otimes 1,\wh{V}\otimes S)$.

\begin{remark}
We see that in our exact sequence, we can think of the quotient $\calA_\phi$ as 
representing our `bulk algebra' as it can be derived from a magnetic 
Hamiltonian on $\ell^2(\Z^2)$ as in \cite{CM96}. Our ideal $C^*(\wh{U})\otimes\calK$ 
can be interpreted as representing the `boundary algebra'. To see 
this we put a boundary on our system so that for the full system the 
Hilbert space is $\calH = \ell^2(\Z\times\N)$, while
 $C^*(\wh{U})$ acts on the boundary $\ell^2(\Z)$, (this action 
being  describable in terms of the bilateral shift operator). 
Tensoring by the compacts in the direction perpendicular to the 
boundary has a physical interpretation as looking at observables 
acting on $\ell^2(\Z\times\N)$ that act on the boundary and decay sufficiently fast away
from it. We would intuitively think of the Hall current of such a system to 
be concentrated at the boundary with a fast decay into the interior, 
so our boundary model lines up with this intuitive picture.
\end{remark}

We now recall some basic definitions from Kasparov theory; the reader 
may consult~\cite{Blackadar, Kasparov80} for a more complete overview. 
A right $C^*$-$A$-module is a space $\calE$ with a right action by a 
$C^*$-algebra $A$ and map $(\,\cdot\mid\cdot\,)_A:\calE\times\calE\to A$, 
which we think of as an $A$-valued inner-product that is compatible with 
the right-action of $A$. We denote the set of adjointable operators on 
$\calE$ with respect to this inner product by $\End_A(\calE)$. Within this 
space are the rank-$1$ endomorphisms, $\Theta_{e,f}$, where 
$\Theta_{e,f}(g) = e\cdot(f|g)_A$ for $e,f,g\in\calE$, which generate the 
finite-rank endomorphsims $\End_A^{00}(\calE)$. The compact 
endomorphisms $\End_A^0(\calE)$ are the closure of the 
finite-rank operators in the operator norm of $\End_A(\calE)$.   
\begin{defn}
Given $\Z_2$-graded $C^*$-algebras $A$ and $B$, an even 
unbounded Kasparov $A$-$B$-module $\left(A,\calE_A,D\right)$ is given by
\begin{enumerate}
\item A $\Z_2$-graded, countably generated, right $C^*$-$B$-module $\calE_B$;
\item A $\Z_2$-graded $*$-homomorphism $\phi\colon A\to\End_B(\calE)$;
\item A self-adjoint, regular, odd operator $D\colon\Dom D\subset\calE\to\calE$ 
such that the graded commutator $[D,\phi(a)]_\pm$ is an adjointable endomorphism, and 
$\phi(a)(1+D^2)^{-1/2}$ is a compact endomorphism for all $a$ in a dense subalgebra $\calA$ of $A$. 
\end{enumerate}
If the module and algebras are trivially graded, then the Kasparov module is called odd.
\end{defn}

We can always pass from unbounded modules to bounded Kasparov 
modules via the mapping $D\mapsto D(1+D^2)^{-1/2}$~\cite{BJ83}.

\subsection{Constructing the Kasparov module}  \label{subsec:constructing_the_right_A_module}
In the last section, we introduced the short exact sequence
\begin{equation} \label{eq:Toeplitz_extension_of_rotation_algebra}
   0 \to C^*(\wh{U})\otimes \calK \xrightarrow{\psi} \calT \to \calA_\phi \to 0. 
\end{equation}
We know that this sequence gives rise to a class in $KK$-theory using 
$\mathrm{Ext}$ groups, but in order to compute the Kasparov product, 
it is desirable to have an explicit Kasparov module that represents a 
class in $KK^1(\calA_\phi, C^*(\wh{U})\otimes \calK) \cong KK^1(\calA_\phi,C^*(\wh{U}))$.

To do this, we introduce our main technical innovation, 
a singular functional $\Psi$ on the subalgebra $C^*(S)$ of $\calT$, which is given by
$$  
\Psi(T) = \res_{s=1} \sum_{k=0}^\infty \langle e_k,Te_k\rangle (1+k^2)^{-s/2},  
$$
where $\{e_k\}$ is any basis of $\ell^2(\N)$.
\begin{prop} 
\label{prop:Psi_vanishes_on_compacts_and_forces_V_powers_to_be_equal}
The functional $\Psi$ is a well-defined trace on $C^*(S)$ such that 
$\Psi\left(S^{l_2}(S^*)^{l_1}S^{n_1}(S^*)^{n_2}\right) = \delta_{l_1-l_2,n_1-n_2}$, 
where $\delta_{a,b}$ is the Kronecker delta. Moreover, $\Psi(T)=0$ for any compact $T$. 
\end{prop}
\begin{proof}
That $\Psi$ is a trace is straightforward from its definition and the 
properties of the usual trace and complex residues.  Thus,
for $S^\alpha(S^*)^\beta\in C^*(S)$, we  see that
$$  
\langle e_k,S^\alpha (S^*)^\beta e_k \rangle 
= \delta_{\alpha,\beta}\langle (S^*)^\alpha e_k, (S^*)^\alpha e_k \rangle 
= \delta_{\alpha,\beta} \chi_{[k,\infty)}(\alpha), 
$$
where $\chi_{[k,\infty)}$ is the indicator function. Hence
\begin{align*}
  \Psi\!\left[S^\alpha (S^*)^\beta\right] 
  &= \res_{s=1}\sum_{k=0}^\infty \delta_{\alpha,\beta}\chi_{[k,\infty)}(\alpha)(1+k^2)^{-s/2} \\
     &= \res_{s=1}\sum_{k=\alpha}^\infty \delta_{\alpha,\beta}(1+k^2)^{-s/2} 
     = \delta_{\alpha,\beta}.
\end{align*}
Similarly $\Psi\!\left((S^*)^\alpha S^\beta\right) = \delta_{\alpha,\beta}$. 
From this we have that, for $l_1\geq n_1$,
$$  
\Psi\left(S^{l_2}(S^*)^{l_1}S^{n_1}(S^*)^{n_2}\right) 
= \Psi\left( S^{l_2}(S^*)^{l_1-n_1+n_2}\right) = \delta_{l_2,l_1-n_1+n_2} = \delta_{l_1-l_2,n_1-n_2}; 
$$
or, for $l_1\leq n_1$, 
$$  
\Psi\left(S^{l_2}(S^*)^{l_1}S^{n_1}(S^*)^{n_2}\right) 
= \Psi\left( S^{l_2-l_1+n_1}(S^*)^{n_2}\right) = \delta_{l_2-l_1+n_1,n_2} = \delta_{l_1-l_2,n_1-n_2}. 
$$

Since $(S^*)^\alpha S^\alpha = 1_{C^*(S)}$, one now readily checks that
\begin{equation}  
\label{eq:Psi_bdd_on_C(S)}
  \Psi(T) \leq \|T\|\,\Psi\!\left(1_{C^*(S)}\right) = \|T\|
\end{equation}
for all $T\in C^*(S)$ and so $\Psi$ extends by continuity to $C^*(S)$. 
For any finite-rank operator, $F\in C^*(S)$, $\langle e_k, Fe_k\rangle \neq 0$ 
for finitely many $k$. This tells us that $\sum_k \langle e_k,F e_k\rangle (1+k^2)^{-s/2}$ 
is holomorphic at $s=1$, whence $\Psi(F)=0$. By \eqref{eq:Psi_bdd_on_C(S)}, $\Psi$ 
vanishes on all the compacts operators on $\ell^2(\N)$.
\end{proof}

In order to simplify computations, we realise $\calT$ as the norm 
closure of the linear span of the operators
$$    
(\wh{V}\otimes S)^{n_1} [(\wh{V}\otimes S)^*]^{n_2}(\wh{U}\otimes 1)^m 
= \wh{V}^{n_1-n_2}\wh{U}^m\otimes S^{n_1}(S^*)^{n_2} 
$$
for $m\in\Z$ and $n_1,n_2\in \N$. We put the $\wh{U}$ on the 
right as we are going to construct a right $C^*(\wh{U})$-module using this presentation.

The first step is the inner product: $(\,\cdot\mid\cdot\,):\calT\times\calT\to C^*(\wh{U})$ given by
\begin{align*}
  &\left(\left. \wh{V}^{l_1-l_2}\wh{U}^{m_1}\otimes S^{l_1}(S^*)^{l_2}\right
  \vert \wh{V}^{n_1-n_2}\wh{U}^{m_2}\otimes S^{n_1}(S^*)^{n_2} \right) \\
    &\hspace{4cm} := \left(\wh{V}^{l_1-l_2}\wh{U}^{m_1}\right)^{\!*} \wh{V}^{n_1-n_2}\wh{U}^{m_2}\, 
    \Psi\!\left[\left(S^{l_1}(S^*)^{l_2}\right)^* S^{n_1}(S^*)^{n_2}\right].
\end{align*}
To show this actually takes values in $C^*(\wh{U})$, we use Proposition \ref{prop:Psi_vanishes_on_compacts_and_forces_V_powers_to_be_equal} to compute that 
\begin{align*}
  \left(\left. \wh{V}^{l_1-l_2}\wh{U}^{m_1}\otimes S^{l_1}(S^*)^{l_2}\right
  \vert \wh{V}^{n_1-n_2}\wh{U}^{m_2}\otimes S^{n_1}(S^*)^{n_2} \right)  
  &=  \wh{U}^{-m_1}\wh{V}^{l_2-l_1}\wh{V}^{n_1-n_2}\wh{U}^{m_2} \delta_{l_1-l_2,n_1-n_2} \\
   &= \wh{U}^{m_2-m_1} \delta_{l_1-l_2,n_1-n_2},
\end{align*}
which is in $C^*(\wh{U})$. With this in mind we construct, in the next result,  a right $C^*(\wh{U})$ module.
\begin{prop}
The map $(\,\cdot\mid\cdot\,):\calT\times\calT \to C^*(\wh{U})$ together 
with an action by right multiplication makes $\calT$ a right $C^*(\wh{U})$-inner-product 
module. Quotienting by vectors of zero length and completing yields a right $C^*(\wh{U})$-module.
\end{prop}
\begin{proof}
Using the equation
$$  
\left(\left. \wh{V}^{l_1-l_2}\wh{U}^{m_1}\otimes S^{l_1}(S^*)^{l_2}\right
\vert \wh{V}^{n_1-n_2}\wh{U}^{m_2}\otimes S^{n_1}(S^*)^{n_2} \right)  
= \wh{U}^{m_2-m_1} \delta_{l_1-l_2,n_1-n_2}  
$$
most of the requirements for $(\,\cdot\mid\cdot\,)$ to be a 
$C^*(\wh{U})$-valued inner-product follow in a straightforward way. We will 
check compatibility with multiplication on the right by elements of $C^*(\wh{U})$. We compute that
\begin{align*}
  &\left(\left. \wh{V}^{l_1-l_2}\wh{U}^{m_1}\otimes S^{l_1}(S^*)^{l_2}\right
  \vert \left(\wh{V}^{n_1-n_2}\wh{U}^{m_2}\otimes S^{n_1}(S^*)^{n_2}\right)
  \cdot( \wh{U}^\alpha\otimes 1) \right) \\
    &\hspace{4cm} = \left(\left. \wh{V}^{l_1-l_2}\wh{U}^{m_1}\otimes S^{l_1}(S^*)^{l_2}\right
    \vert \wh{V}^{n_1-n_2}\wh{U}^{m_2+\alpha}\otimes S^{n_1}(S^*)^{n_2}\right) \\
   &\hspace{4cm} =  \wh{U}^{m_2-m_1+\alpha} \delta_{l_1-l_2,n_1-n_2}  \\
    &\hspace{4cm} = \left(\wh{U}^{m_2-m_1}\delta_{l_1-l_2,n_1-n_2}\right)\wh{U}^\alpha  \\
    &\hspace{4cm} = \left(\left. \wh{V}^{l_1-l_2}\wh{U}^{m_1}\otimes S^{l_1}(S^*)^{l_2}\right
    \vert \wh{V}^{n_1-n_2}\wh{U}^{m_2}\otimes S^{n_1}(S^*)^{n_2} \right) \cdot \wh{U}^\alpha
\end{align*}
for $\alpha\in\Z$. Obtaining the result for arbitrary elements in 
$C^*(\wh{U})$ is a simple extension of this.
\end{proof}
We denote our $C^*$-module by $Z_{C^*(\wh{U})}$ and 
inner-product by $(\,\cdot\mid\cdot\,)_{C^*(\wh{U})}$. The 
point of the singular trace $\Psi$
becomes apparent in the next proposition where
we construct a left action of $\calA_\phi$ on $Z_{C^*(\wh{U})}$.

\begin{prop} 
\label{prop:repn_of_A_on_right_A_module}
There is an adjointable representation if $\calA_\phi$ on $Z_{C^*(\wh{U})}$.
\end{prop}
\begin{proof}
Clearly we can multiply elements of $Z_{C^*(\wh{U})}$ by $\calT$ on the left, but by 
Proposition \ref{prop:Psi_vanishes_on_compacts_and_forces_V_powers_to_be_equal}, 
we know that $(\wh{U}^j\wh{V}^k\otimes k)\cdot Z_{C^*(\wh{U})} = 0$ if $k\in\calK$. 
Therefore the representation of $\calT$ descends to a representation of 
$\calT/\psi[C(S^1)\otimes\calK] \cong \calA_\phi$. This gives us the explicit left-action by 
\begin{align*}
   (\wh{U}^\alpha \wh{V}^\beta)\cdot\left(\wh{V}^{n_1-n_2}\wh{U}^m \otimes S^{n_1}(S^*)^{n_2}\right) 
   &=  (\wh{U}^\alpha \wh{V}^\beta \wh{V}^{n_1-n_2}\wh{U}^m) \otimes S^{n_1+\beta}(S^*)^{n_2} \\
     &= e^{2\pi i\phi\alpha(n_1-n_2+\beta)} \wh{V}^{\beta+ n_1-n_2}\wh{U}^{m+\alpha} 
     \otimes S^{\beta + n_1}(S^*)^{n_2}
\end{align*}
for $\alpha,\beta\in\Z$ with $\beta \geq 0$ and
$$   
(\wh{U}^\alpha \wh{V}^\beta)\cdot\left(\wh{V}^{n_1-n_2}\wh{U}^m \otimes S^{n_1}(S^*)^{n_2}\right) 
= e^{2\pi i\phi\alpha(n_1-n_2+\beta)} \wh{V}^{\beta+ n_1-n_2}\wh{U}^{m+\alpha} 
\otimes S^{\beta}(S^*)^{n_2+|\beta|}  
$$
for $\beta<0$. It follows that, as operators on 
$Z_{C^*(\wh{U})}$, $\wh{U}\wh{V} = e^{2\pi i\phi}\wh{V}\wh{U}$. 
Next we just need to verify that the action is adjointable as a module over 
$C^*(\wh{U})$. For this it suffices to check that multiplication by 
$\wh{U}$ and $\wh{V}$ are adjointable. We compute that 
\begin{align*}
  &\left(\left. \wh{U}\cdot\left(\wh{V}^{l_1-l_2}\wh{U}^{m_1}\otimes S^{l_1}(S^*)^{l_2}\right)\right
  \vert \wh{V}^{n_1-n_2}\wh{U}^{m_2}\otimes S^{n_1}(S^*)^{n_2} \right)_{C^*(\wh{U})} \\
   &\hspace{4cm}
   = \left(\left. e^{2\pi i\phi(l_1-l_2)}\wh{V}^{l_1-l_2}\wh{U}^{m_1+1} \otimes S^{l_1}(S^*)^{l_2}\right
   \vert \wh{V}^{n_1-n_2}\wh{U}^{m_2}\otimes S^{n_1}(S^*)^{n_2} \right)_{C^*(\wh{U})} \\
  &\hspace{4cm}= e^{-2\pi i\phi(l_1-l_2)}\wh{U}^{m_2-1-m_1}\delta_{l_1-l_2,n_1-n_2} \\
  &\hspace{4cm}= \left(\left.\wh{V}^{l_1-l_2}\wh{U}^{m_1}\otimes S^{l_1}(S^*)^{l_2}\right
  \vert e^{-2\pi i\phi(n_1-n_2)} \wh{V}^{n_1-n_2}\wh{U}^{m_2-1}
  \otimes S^{n_1}(S^*)^{n_2} \right)_{C^*(\wh{U})} \\
  &\hspace{4cm}= \left(\left. \wh{V}^{l_1-l_2}\wh{U}^{m_1}\otimes S^{l_1}(S^*)^{l_2}\right
  \vert \wh{U}^{-1} \cdot\left(\wh{V}^{n_1-n_2}\wh{U}^{m_2}
  \otimes S^{n_1}(S^*)^{n_2}\right) \right)_{C^*(\wh{U})}
\end{align*}
and then
\begin{align*}
   &\left(\left. \wh{V}\cdot\left(\wh{V}^{l_1-l_2}\wh{U}^{m_1}\otimes S^{l_1}(S^*)^{l_2}\right)\right
   \vert \wh{V}^{n_1-n_2}\wh{U}^{m_2}\otimes S^{n_1}(S^*)^{n_2} \right)_{C^*(\wh{U})} \\
    &\hspace{4cm}= \left(\left.\wh{V}^{l_1-l_2+1}\wh{U}^{m_1}\otimes S^{l_1+1}(S^*)^{l_2}\right
    \vert \wh{V}^{n_1-n_2}\wh{U}^{m_2}\otimes S^{n_1}(S^*)^{n_2} \right)_{C^*(\wh{U})} \\
     &\hspace{4cm}= \wh{U}^{m_2-m_1}\delta_{l_1-l_2+1,n_1-n_2} \\
     &\hspace{4cm}= \wh{U}^{m_2-m_1}\delta_{l_1-l_2,n_1-n_2-1} \\
     &\hspace{4cm}= \left(\left.\wh{V}^{l_1-l_2}\wh{U}^{m_1}\otimes S^{l_1}(S^*)^{l_2}\right
     \vert \wh{V}^{n_1-n_2-1}\wh{U}^{m_2}\otimes S^{n_1}(S^*)^{n_2+1} \right)_{C^*(\wh{U})} \\
     &\hspace{4cm}= \left(\left.\wh{V}^{l_1-l_2}\wh{U}^{m_1}\otimes S^{l_1}(S^*)^{l_2}\right
     \vert \wh{V}^{-1}\cdot\left( \wh{V}^{n_1-n_2}\wh{U}^{m_2}
     \otimes S^{n_1}(S^*)^{n_2}\right) \right)_{C^*(\wh{U})}.
\end{align*}
and so our generating elements are adjointable and unitary on the dense span of monomials
in $Z_{C^*(\wh{U})}$. Thus if $\wh{U},\,\wh{V}$ are bounded, they will generate an 
adjointable representation of $\calA_\phi$. 
To consider the boundedness of $\wh{U}$ and $\wh{V}$, we first note 
that the inner-product in $Z_{C^*(\wh{U})}$ is defined 
from multiplication in $\calT$ and the functional $\Psi$, 
which has the property $\Psi(T)\leq \|T\|$, by Equation 
\eqref{eq:Psi_bdd_on_C(S)}. These observations imply that
$$ 
\|a\|_{\End(Z)} 
= \sup_{\substack{ z\in Z \\ \|z\|=1}} (a\cdot z\mid a\cdot z)_{C^*(\wh{U})}
\leq \sup_{\substack{ z\in Z \\ \|z\|=1}} \|aa^*\|\, (z\mid z)_{C^*(\wh{U})} 
= \|aa^*\|.
$$
Therefore the action of $\calA_\phi$ is bounded, and
so extends to an adjointable action on $Z_{C^*(\wh{U})}$.
\end{proof}


In Section \ref{subsec:left_module_over_Aop}, we show that by considering a 
left module ${}_{C^*(\wh{U})}Z$, we may also obtain an adjointable 
representation of $\calA_\phi^\mathrm{op}$.
Before we finish building our Kasparov module, we need some further 
results arising from properties of the singular trace $\Psi$.
\begin{prop} 
\label{prop:finite_powers_of_V_dont_make_a_difference_in_Psi_metric}
Let $l_1-l_2=n_1-n_2$. Then $\wh{V}^{n_1-n_2}\wh{U}^m\otimes S^{n_1}(S^*)^{n_2} 
= \wh{V}^{l_1-l_2}\wh{U}^m\otimes S^{l_1}(S^*)^{l_2}$ as elements in $Z_{C^*(\wh{U})}$.
\end{prop}
\begin{proof}
We can assume without loss of generality that $l_1=n_1+k$ and 
$l_2=n_2+k$ for some $k\in\Z$. As a preliminary, we compute 
$\Psi\!\left[(S^{n_1}(S^*)^{n_2}-S^{n_1+k}(S^*)^{n_2+k})^* (S^{n_1}(S^*)^{n_2}
-S^{n_1+k}(S^*)^{n_2+k})\right]$. Firstly we expand
\begin{align*}
  &\left(S^{n_1}(S^*)^{n_2}-S^{n_1+k}(S^*)^{n_2+k}\right)^* 
  \left(S^{n_1}(S^*)^{n_2}-S^{n_1+k}(S^*)^{n_2+k}\right) \\
    &\hspace{4cm} = S^{n_2}(S^*)^{n_1}S^{n_1}(S^*)^{n_2} - S^{n_2}(S^*)^{n_1}S^{n_1+k}(S^*)^{n_2+k} \\
    &\hspace{5cm} - S^{n_2+k}(S^*)^{n_1+k}S^{n_1}(S^*)^{n_2} 
    +  S^{n_2+k}(S^*)^{n_1+k}S^{n_1+k}(S^*)^{n_2+k} \\
    &\hspace{4cm} = S^{n_2}(S^*)^{n_2} - S^{n_2+k}(S^*)^{n_2+k} 
    - S^{n_2+k}(S^*)^{n_2+k} + S^{n_2+k}(S^*)^{n_2+k} \\
    &\hspace{4cm} = S^{n_2}(S^*)^{n_2} - S^{n_2+k}(S^*)^{n_2+k}.
\end{align*}
We now recall that $\Psi(S^{\alpha}(S^*)^{\beta}) = \delta_{\alpha,\beta}$, so that 
\begin{align*}
  &\Psi\!\left[(S^{n_1}(S^*)^{n_2}-S^{n_1+k}(S^*)^{n_2+k})^* 
  (S^{n_1}(S^*)^{n_2}-S^{n_1+k}(S^*)^{n_2+k})\right] \\
   &\hspace{7cm} = \Psi(S^{n_2}(S^*)^{n_2}) - \Psi(S^{n_2+k}(S^*)^{n_2+k}) = 0.
\end{align*}
From this point, it is a simple task to show that 
$\wh{V}^{n_1-n_2}\wh{U}^m\otimes S^{n_1}(S^*)^{n_2} 
= \wh{V}^{n_1-n_2}\wh{U}^m\otimes S^{n_1+k}(S^*)^{n_2+k}$ in the norm 
induced by $(\,\cdot\mid\cdot\,)_{C^*(\wh{U})}$.
\end{proof}

\begin{lemma} 
\label{lemma:projections_relatively_orthogonal_in_vhat_direction}
Let $e_{n_1,n_2,m}$ denote the element 
$\wh{V}^{n_1-n_2}\wh{U}^m\otimes S^{n_1}(S^*)^{n_2}\in Z_{C^*(\wh{U})}$. Then for all $k\in\Z$
$$  	
\Theta_{e_{l_1,l_2,k},e_{l_1,l_2,k}}(e_{n_1,n_2,m}) = \delta_{l_1-l_2,n_1-n_2}\, e_{n_1,n_2,m}, 
$$
where $\Theta_{e,f}(g) = e(f|g)_{C^*(\wh{U})}$ are the rank-$1$ 
endomorphisms that generate $\End_{C^*(\wh{U})}^{0}(Z)$.
\end{lemma}
\begin{proof}
We check that 
\begin{align*}
   \Theta_{e_{l_1,l_2,k},e_{l_1,l_2,k}}(e_{n_1,n_2,m}) 
   &=  \wh{V}^{l_1-l_2}\wh{U}^k\otimes S^{l_1}(S^*)^{l_2} \\
     &\qquad {} \times \left(\left.  \wh{V}^{l_1-l_2}\wh{U}^k\otimes S^{l_1}(S^*)^{l_2}\right
     \vert \wh{V}^{n_1-n_2}\wh{U}^{m}\otimes S^{n_1}(S^*)^{n_2} \right)_{C^*(\wh{U})} \\
     &= \left( \wh{V}^{l_1-l_2}\wh{U}^k\otimes S^{l_1}(S^*)^{l_2}\right)
     \cdot \wh{U}^{m-k} \delta_{l_1-l_2,n_1-n_2} \\
     &= \wh{V}^{n_1-n_2}\wh{U}^m \otimes S^{n_1}(S^*)^{n_2}\delta_{l_1-l_2,n_1-n_2},
\end{align*}
where we have used Proposition \ref{prop:finite_powers_of_V_dont_make_a_difference_in_Psi_metric}.
\end{proof}

With these preliminary results out the way, we now state the main result of this subsection.
\begin{prop} 
\label{prop:we_get_right_kas_mod_for_A}
Define the operator $N:\Dom(N)\subset Z \to Z$ such that 
$N\left(\wh{V}^{n_1-n_2}\wh{U}^m\otimes S^{n_1}(S^*)^{n_2}\right) 
= (n_1-n_2)\wh{V}^{n_1-n_2}\wh{U}^m\otimes S^{n_1}(S^*)^{n_2}$. 
Then $\left({\calA_\phi},Z_{C^*(\wh{U})}, N\right)$ is an unbounded, odd Kasparov module.
\end{prop}
\begin{proof}
Lemma  \ref{lemma:projections_relatively_orthogonal_in_vhat_direction} shows that
for any $n_1,n_2$ with $n_1-n_2=k$, the operator
$\Phi_{k}=\Theta_{e_{n_1,n_2,0},e_{n_1,n_2,0}}$ is an 
adjointable projection. These projections form an orthogonal family 
$$  
\Phi_{l}\Phi_{k} = \delta_{l,k}\Phi_k 
$$
by Lemma \ref{lemma:projections_relatively_orthogonal_in_vhat_direction}, and 
it is straightforward to show that $\sum_{k\in\Z}\Phi_k$ is the identity of $Z$ (convergence in the
strict topology). 
The arguments used in~\cite{PR06} show that given $z\in Z$ and 
defining $\Phi_k z = z_k$, we have that
$$  
z = \sum_{k\in\Z} z_k.  
$$
This allows us to define a number operator 
$$  
Nz = \sum_{k\in\Z} k z_k 
$$
for those $z\in\Dom(N)=\left\{ \sum_k z_k\,:\, \sum_k k^2(z_k|z_k)_{C^*(\wh{U})} < \infty \right\}$. 
As $N$ is given in in its spectral representation, standard proofs show that 
$N$ is self-adjoint (again, see~\cite{PR06} for an explicit proof).

To show that $N$ is regular, we observe that 
$$  
N^2\left(\wh{V}^{n_1-n_2}\wh{U}^m\otimes S^{n_1}(S^*)^{n_2}\right) 
= (n_1-n_2)^2\, \wh{V}^{n_1-n_2}\wh{U}^m\otimes S^{n_1}(S^*)^{n_2} 
$$
and so $N^2$ has the spanning set of $\calT$ as eigenvectors. 
Therefore $(1+N^2)$ has dense range and so $N$ is regular. 

To check that we have an unbounded Kasparov module, 
we need to show that $[N,a]$ is a bounded endomorphism for 
$a$ in a dense subset of $\calA_\phi$ and that 
$(1+N^2)^{-1/2}\in\End_{C^*(\wh{U})}^0(Z)$. We have that, for $\beta\geq 0$
\begin{align*}
  N(\wh{U}^\alpha \wh{V}^\beta)\left(\wh{V}^{n_1-n_2}\wh{U}^m \otimes S^{n_1}(S^*)^{n_2}\right) 
  &= N\left( e^{2\pi i\phi\alpha(n_1-n_2+\beta)} \wh{V}^{n_1-n_2+\beta}\wh{U}^{m+\alpha}
  \otimes S^{n_1+\beta}(S^*)^{n_2}\right) \\ 
    &= (n_1-n_2+\beta)e^{2\pi i\phi\alpha(n_1-n_2+\beta)} \wh{V}^{n_1-n_2+\beta}\wh{U}^{m+\alpha}
    \otimes S^{n_1+\beta}(S^*)^{n_2}
\end{align*}
and
$$  
(\wh{U}^\alpha \wh{V}^\beta)N\left(\wh{V}^{n_1-n_2}\wh{U}^m \otimes S^{n_1}(S^*)^{n_2}\right) 
= (n_1-n_2)e^{2\pi i\phi\alpha(n_1-n_2+\beta)} \wh{V}^{n_1-n_2+\beta}\wh{U}^{m+\alpha}
\otimes S^{n_1+\beta}(S^*)^{n_2}, 
$$
which implies that $[N,\wh{U}^\alpha\wh{V}^\beta] = \beta \wh{U}^\alpha \wh{V}^\beta$ 
since the span of $\wh{V}^{n_1-n_2}\wh{U}^m\otimes S^{n_1}(S^*)^{n_2}$ is 
dense in the domain of $N$ in the graph norm. Hence for an element 
$a=\sum_{\alpha,\beta}a_{\alpha,\beta}\wh{U}^\alpha\wh{V}^\beta$ in a 
dense subset of $\calA_\phi$ with $(a_{\alpha,\beta})\in\calS(\Z^2)$, the Schwartz class sequences,
we have that 
$$ 
[N,a] = \sum_{\alpha,\beta}\beta a_{\alpha,\beta} \wh{U}^\alpha \wh{V}^\beta 
$$
which is in $\calA_\phi$ as $\beta a_{\alpha,\beta}\in\calS(\Z^2)$ and 
therefore is bounded. An entirely analogous argument also works for $\beta < 0$.

Finally, we recall that $N^2$ has a set of eigenvectors given by the spanning functions \\ 
$\left\{\wh{V}^{n_1-n_2}\wh{U}^m\otimes S^{n_1}(S^*)^{n_2}\,:\,n_1,n_2\in\N,\,m\in\Z\right\}$. 
This means that we can write
$$  
N^2 = \bigoplus_{k\in\Z} k^2 \Phi_{k} 
$$
where $\Phi_{k}$ is the projection on onto 
${\rm span}\{\wh{V}^{n_1-n_2}\wh{U}^m\otimes S^{n_1}(S^*)^{n_2}\in Z_{C^*(\wh{U})}:\,
n_1-n_2=k,\,m\in\Z\}$. 
As the projections $\Phi_{k}$ can be written as a rank one operator
$\Theta_{e_{n_1,n_2,0},e_{n_1,n_2,0}} \in \End_{C^*(\wh{U})}^{00}(Z)$, we have that
$$  
(1+ N^2)^{-1/2} = \bigoplus_{k\in\Z} \left(1+ k^2\right)^{-1/2} \Phi_{k}
$$
is a norm-convergent sum of elements in $\End_{C^*(\wh{U})}^{00}(Z)$ 
and is therefore in $\End_{C^*(\wh{U})}^{0}(Z)$.
\end{proof}

\subsection{A left module with $\calA_\phi^\mathrm{op}$-action} \label{subsec:left_module_over_Aop}
The module $Z_{C^*(\wh{U})}$ has more structure. It is in fact a left 
$C^*$-module over $C^*(\wh{U})$ where we  define an inner-product by
\begin{align*}
   {}_{C^*(\wh{U})}\left(\left.\wh{V}^{l_1-l_2}\wh{U}^{m_1}\otimes S^{l_1}(S^*)^{l_2}\right
   \vert \wh{V}^{n_1-n_2}\wh{U}^{m_2}\otimes S^{n_1}(S^*)^{n_2}\right) 
   &= \wh{V}^{l_1-l_2}\wh{U}^{m_1}\left(\wh{V}^{n_1-n_2}\wh{U}^{m_2}\right)^*  \\
      &\qquad {} \times \Psi\!\left[S^{l_1}(S^*)^{l_2}\left(S^{n_1}(S^*)^{n_2}\right)^*\right] \\
   &= \wh{V}^{l_1-l_2}\wh{U}^{m_1-m_2}\wh{V}^{n_2-n_1}\delta_{l_1-l_2,n_1-n_2} \\
   &= \eta_{n_1-n_2}^{-1}(\wh{U}^{m_1-m_2})\delta_{l_1-l_2,n_1-n_2},
\end{align*}
recalling that $\eta_n(\wh{U}^m) = \wh{V}^{-n}\wh{U}^m\wh{V}^n$ is the 
automorphism defining the crossed-product structure. We check 
compatibility of ${}_{C^*(\wh{U})}(\,\cdot\mid\cdot\,)$ with left-multiplication by $C^*(\wh{U})$, where
\begin{align*}
   &{}_{C^*(\wh{U})}\left(\left.\wh{U}\wh{V}^{l_1-l_2}\wh{U}^{m_1}\otimes S^{l_1}(S^*)^{l_2}\right
   \vert \wh{V}^{n_1-n_2}\wh{U}^{m_2}\otimes S^{n_1}(S^*)^{n_2}\right)  \\
   &\hspace{5.5cm} = \wh{U}\wh{V}^{l_1-l_2}\wh{U}^{m_1-m_2}\wh{V}^{n_1-n_1}\delta_{l_1-l_2,n_1-n_2} \\
     &\hspace{5.5cm} = 
     \wh{U} \cdot {}_{C^*(\wh{U})}\left(\left.\wh{V}^{l_1-l_2}\wh{U}^{m_1}\otimes S^{l_1}(S^*)^{l_2}\right
     \vert \wh{V}^{n_1-n_2}\wh{U}^{m_2}\otimes S^{n_1}(S^*)^{n_2}\right).
\end{align*}
The other axioms for a left $C^*(\wh{U})$-valued inner-product are 
straightforward. We complete in the induced norm and denote our 
left-module by ${}_{C^*(\wh{U})}Z$.

\begin{prop} 
\label{prop:adjointable_rep_of_Aop_on_left_module}
There is an adjointable representation of 
$\calA_{-\phi}\cong\calA_\phi^\mathrm{op}$ on ${}_{C^*(\wh{U})}Z$.
\end{prop}
\begin{proof}
We construct an action by $C^*(U,V)\cong\calA_\phi^{\mathrm{op}}$ by defining
\begin{align*}
  U\cdot\left(\wh{V}^{n_1-n_2}\wh{U}^m\otimes S^{n_1}(S^*)^{n_2}\right) 
  &= \left(\wh{V}^{n_1-n_2}\wh{U}^m\otimes S^{n_1}(S^*)^{n_2}\right)\cdot \wh{U} 
  = \wh{V}^{n_1-n_2}\wh{U}^{m+1}\otimes S^{n_1}(S^*)^{n_2}, \\
  V\cdot\left(\wh{V}^{n_1-n_2}\wh{U}^m\otimes S^{n_1}(S^*)^{n_2}\right) 
  &= \left(\wh{V}^{n_1-n_2}\wh{U}^m\otimes S^{n_1}(S^*)^{n_2}\right)\cdot \wh{V} 
  = e^{2\pi i\phi m}\wh{V}^{n_1-n_2+1}\wh{U}^m\otimes S^{n_1+1}(S^*)^{n_2}
\end{align*}
and extending  to the whole algebra. One finds that, as 
operators on ${}_{C^*(\wh{U})}Z$, $UV = e^{-2\pi i\phi}VU$. 
As previously, we check adjointability on generating elements, where
\begin{align*}
    &{}_{C^*(\wh{U})}\left(\left.U\cdot\left(\wh{V}^{l_1-l_2}\wh{U}^{m_1}\otimes S^{l_1}(S^*)^{l_2}\right)\right
    \vert \wh{V}^{n_1-n_2}\wh{U}^{m_2}\otimes S^{n_1}(S^*)^{n_2}\right) \\
     &\hspace{5cm}= \eta_{n_1-n_2}^{-1}(\wh{U}^{m_1+1-m_2})\delta_{n_1-n_2,l_1-l_2} \\
      &\hspace{5cm}=  {}_{C^*(\wh{U})}\left(\left.\wh{V}^{l_1-l_2}\wh{U}^{m_1}\otimes S^{l_1}(S^*)^{l_2}\right
      \vert \wh{V}^{n_1-n_2}\wh{U}^{m_2-1}\otimes S^{n_1}(S^*)^{n_2}\right) \\
      &\hspace{5cm}=  {}_{C^*(\wh{U})}\left(\left.\wh{V}^{l_1-l_2}\wh{U}^{m_1}\otimes S^{l_1}(S^*)^{l_2}\right
      \vert U^{-1}\cdot\left( \wh{V}^{n_1-n_2}\wh{U}^{m_2}\otimes S^{n_1}(S^*)^{n_2}\right) \right)
\end{align*}
as expected. For $V$, we find that
\begin{align*}
   &{}_{C^*(\wh{U})}\left(\left.V\cdot\left(\wh{V}^{l_1-l_2}\wh{U}^{m_1}\otimes S^{l_1}(S^*)^{l_2}\right)\right
   \vert \wh{V}^{n_1-n_2}\wh{U}^{m_2}\otimes S^{n_1}(S^*)^{n_2}\right)  \\
    &\hspace{4cm} =  
    {}_{C^*(\wh{U})}\left(\left.e^{2\pi i\phi m_1}\wh{V}^{l_1-l_2+1}\wh{U}^{m_1}
    \otimes S^{l_1+1}(S^*)^{l_2}\right\vert \wh{V}^{n_1-n_2}\wh{U}^{m_2}\otimes S^{n_1}(S^*)^{n_2}\right) \\
    &\hspace{4cm} = 
    e^{2\pi i\phi m_1}\wh{V}^{l_1-l_2+1}\wh{U}^{m_1-m_2}\wh{V}^{n_2-n_1}\delta_{l_1-l_2+1,n_1-n_2} \\
    &\hspace{4cm} = 
    e^{2\pi i\phi m_1}e^{-2\pi i\phi(m_1-m_2)}\wh{V}^{l_1-l_2}\wh{U}^{m_1-m_2}\wh{V}^{n_2-n_1+1} 
    \delta_{l_1-l_2,n_1-n_2-1} \\
    &\hspace{4cm} = 
    {}_{C^*(\wh{U})}\left(\left.\wh{V}^{l_1-l_2}\wh{U}^{m_1}\otimes S^{l_1}(S^*)^{l_2}\right
    \vert e^{-2\pi i\phi m_2}\wh{V}^{n_1-n_2-1}\wh{U}^{m_2}\otimes S^{n_1}(S^*)^{n_2+1}\right) \\
    &\hspace{4cm} = {}_{C^*(\wh{U})}\left(\left.\wh{V}^{l_1-l_2}\wh{U}^{m_1}\otimes S^{l_1}(S^*)^{l_2}\right
    \vert (\wh{V}^{n_1-n_2}\wh{U}^{m_2}\otimes S^{n_1}(S^*)^{n_2})\wh{V}^{-1}\right) \\
    &\hspace{4cm} = {}_{C^*(\wh{U})}\left(\left.\wh{V}^{l_1-l_2}\wh{U}^{m_1}\otimes S^{l_1}(S^*)^{l_2}\right
    \vert V^{-1}\cdot\left(\wh{V}^{n_1-n_2}\wh{U}^{m_2}\otimes S^{n_1}(S^*)^{n_2}\right)\right)
\end{align*}
and so our generating elements are adjointable and unitary on the dense span of monomials
in ${}_{C^*(\wh{U})}Z$. Thus if $U,\,V$ are bounded, they will generate an 
adjointable representation of $\calA_\phi^\mathrm{op}$. 
To consider the boundedness of $U$ and $V$, we first note 
that the inner-product in ${}_{C^*(\wh{U})}Z$ is defined 
from multiplication in $\calT$ and the functional $\Psi$, 
which has the property $\Psi(T)\leq \|T\|$, by Equation 
\eqref{eq:Psi_bdd_on_C(S)}. These observations imply that
$$ 
\|a^{\mathrm{op}}\|_{\End(Z)} = \sup_{\substack{ z\in Z \\ \|z\|=1}} {}_{C^*(\wh{U})}(a^\mathrm{op}\cdot z\mid a^\mathrm{op}\cdot z) \leq \sup_{\substack{ z\in Z \\ \|z\|=1}} \|a^\mathrm{op}(a^\mathrm{op})^*\|\, {}_{C^*(\wh{U})}(z\mid z) = \|a^\mathrm{op}(a^\mathrm{op})^*\|. 
$$
Therefore the action of $\calA_\phi^\mathrm{op}$ is bounded, and
so extends to an adjointable action on ${}_{C^*(\wh{U})}Z$.
\end{proof}

\begin{remark}
Our construction of ${}_{C^*(\wh{U})}Z$ shows that $Z$ can be equipped with a 
bimodule structure over $C^*(\wh{U})$. Proposition \ref{prop:repn_of_A_on_right_A_module} and \ref{prop:adjointable_rep_of_Aop_on_left_module} show that the right (resp. left) 
module comes with an adjointable representation of $\calA_\phi$ (resp. $\calA_\phi^\mathrm{op}$). 
While it may be tempting to think so, we emphasise that these representations 
are \emph{not} adjointable on the left (resp. right) module.

Another thing to note is that the actions of $\calA_\phi$ and $\calA_\phi^\mathrm{op}$ 
on $Z$ commute. The proof of this is a computation; the only part that requires 
some work is to show that $[\wh{U},V]=0$. Since
$$  
\wh{U}V \left(\wh{V}^{n_1-n_2}\wh{U}^m\otimes V^{n_1}(V^*)^{n_2}\right) 
= e^{2\pi i\phi(n_1-n_2+1)}e^{2\pi i\phi m} \wh{V}^{n_1-n_2+1}\wh{U}^{m+1}\otimes V^{n_1}(V^*)^{n_2} 
$$
and 
$$  
V\wh{U} \left(\wh{V}^{n_1-n_2}\wh{U}^m\otimes V^{n_1}(V^*)^{n_2}\right) 
= e^{2\pi i\phi(m+1)}e^{2\pi i\phi(n_1-n_2)} \wh{V}^{n_1-n_2+1}\wh{U}^{m+1}\otimes V^{n_1}(V^*)^{n_2}, 
$$
we find that, as required, $[\wh{U},V]=0$. Once again, we reiterate that these 
actions cannot be considered as simultaneous representations on the level of 
right or left $C^*(\wh{U})$-modules.
\end{remark}

All the technical results in Section \ref{subsec:constructing_the_right_A_module} 
about the singular trace $\Psi$ still hold in the left-module setting. In particular, 
a completely analogous argument to the proof of 
Proposition \ref{prop:we_get_right_kas_mod_for_A} gives us the following.
\begin{prop}
The tuple $\left( \calA_\phi^\mathrm{op}, {}_{C^*(\wh{U})}Z,N\right)$ is an 
odd, unbounded $\calA_\phi^\mathrm{op}$-$C^*(\wh{U})^\mathrm{op}$ Kasparov module.
\end{prop}

\subsection{Relating the module to the extension class}

Now we put the pieces together.
By~\cite[Section 7]{Kasparov80}, the extension class associated to 
$\left(\calA_\phi, Z_{C^*(\wh{U})},N\right)$ comes from the short exact sequence
\begin{equation} 
\label{eq:extension_SES}
  0 \to \End_{C^*(\wh{U})}^0(PZ) \to C^*(P\calA_\phi P) \to \calA_\phi \to 0,
\end{equation}
where $P=\chi_{[0,\infty)}(N)$ is the non-negative spectral projection.

We have that the map $W:Z \to \ell^2(\Z)\otimes C^*(\wh{U})$ given by
$$  
W\left( \wh{V}^{n_1-n_2}\wh{U}^m\otimes S^{n_1}(S^*)^{n_2} \right) 
= e_{n_1-n_2}\otimes \wh{U}^m 
$$
is an adjointable unitary isomorphism. Conjugation by the unitary $W$ gives an explicit 
isomorphism $\End_{C^*(\wh{U})}^0(PZ) \cong \calK[\ell^2(\N)]\otimes C^*(\wh{U})$. 
This isomorphism is compatible with the sequence in 
equation \eqref{eq:extension_SES} in that the commutators $[P,S^k]$ and 
$[P,(S^*)^k]$ generate $\calK[\ell^2(\N)]$. With a suitable identification, the map
$$  
\End_{C^*(\wh{U})}^0(PZ) \xhookrightarrow{\iota} C^*(P\calA_\phi P) 
$$
is just inclusion.

Now define the isomorphism $\zeta:C^*(P \calA_\phi P) \to \calT$ by
\begin{align*}
   &\zeta(P\wh{V}^n P) = (\wh{V}\otimes S)^n,  &&\zeta(P\wh{V}^{-n}P) = [(\wh{V}\otimes S)^*]^n
\end{align*}   
for $n\geq 0$ and
$$  
\zeta(\wh{U}^m) = \sum_{j=0}^\infty (\wh{V}^*)^j \wh{U}^m \wh{V}^j \otimes S^j (1-SS^*) (S^*)^j 
$$
and then extend accordingly. Then we have that the diagram
$$  
  \xymatrix{  0 \ar[r] & \calK\otimes C^*(\wh{U}) \ar[r] & \calT \ar[r] & \calA_\phi \ar[r] & 0 \\
              0 \ar[r] & \End_{C^*(\wh{U})}^0(P\calT) \ar[r] \ar[u]^{\cong}_{\mathrm{Ad}W} 
              & C^*(P\calA_\phi P) \ar[r] \ar[u]^{\cong}_\zeta & \calA_\phi \ar@{=}[u] \ar[r] & 0  }
$$
commutes, and so these extensions are unitarily equivalent. 
We summarise this Section by  the following.
\begin{prop}
The extension class representing the short exact sequence of 
Equation \eqref{eq:Toeplitz_extension_of_rotation_algebra} is the same as the 
class represented by the Kasparov module 
$\left(\calA_\phi,Z_{C^*(\wh{U})},N\right)$ in $KK^1(\calA_\phi,C^*(\wh{U}))$.
\end{prop}


\section{The bulk-edge correspondence and the Kasparov product}

\subsection{Overview of the main result}
Once again recall the short exact sequence
$$     
0 \to C^*(\wh{U})\otimes \calK[\ell^2(\N)] \xrightarrow{\psi} \calT \to \calA_\phi \to 0. 
$$
The ideal is regarded as our boundary data, as we can consider it acting on 
$\ell^2(\Z\times\N)$ but with compact operators acting in the direction 
perpendicular to the boundary. The quotient $\calA_\phi$ describes a 
quantum Hall system in the absence of the boundary.

There is an obvious spectral triple in the work of Bellissard 
et al.~\cite{Bellissard94}\footnote{The authors of~\cite{Bellissard94} 
actually deal with Fredholm modules, 
but there is a very natural extension to the setting of spectral triples.} 
for the 
boundary-free quantum Hall system. We use the notation  
$\left(\calA_{-\phi},\ell^2(\Z^2)\oplus\ell^2(\Z^2),X\right)$ for this triple which represents a 
class in $KK^0(\calA_{-\phi},\C)$. Here we have
$X = \begin{pmatrix} 0 & X_1-iX_2 \\ X_1+iX_2 & 0 \end{pmatrix}$, 
where $X_1$ and $X_2$ are position (or, equivalently, number) operators on $\ell^2(\Z^2)$. 
We think of this as a `Dirac-type' operator.

We also have the natural spectral triple on $C^*(\wh{U})$ that 
gives us a class 
$\left[(C^*(\wh{U}),\ell^2(\Z)_\C, M)\right]\in KK^1(C(S^1),\C)\cong KK^1(C^*(\wh{U})\otimes\calK,\C)$ 
for $M$ the position/number operator on $\ell^2(\Z)$. Our idea is to use the Kasparov 
module that represents the Toeplitz extension to relate the bulk and boundary 
spectral triples via the internal Kasparov product. Namely, we claim that, under the map
$$  
KK^1(\calA_\phi, C(S^1)) \times KK^1(C(S^1), \C) \to KK^0(\calA_\phi,\C), 
$$
we have that
$$   
\left[( {\calA_\phi},Z_{C^*(\wh{U})}, N)\right] \hat{\otimes}_{C^*(\wh{U})} 
\left[(C^*(\wh{U}),\ell^2(\Z)_\C, M)\right] = -\,\left[({\calA_\phi},\ell^2(\Z^2)_\C, X,\Gamma)\right]. 
$$
Of course, our original boundary-free spectral triple is in $K^0(\calA_{-\phi})$, 
not $K^0(\calA_{\phi})$. 
By using the extra structure coming from the 
left-module $\left( \calA_\phi^\mathrm{op}, {}_{C^*(\wh{U})}\calT,N\right)$, 
we are able to resolve this discrepancy and obtain the Bellissard spectral 
triple from the product module up to an explicit unitary equivalence.

\subsection{The details}
\subsubsection{The boundary spectral triple and the product}
We have our module $\beta=\left( {\calA_\phi},Z_{C^*(\wh{U})}, N\right)$ 
giving rise to a class in $KK^1(\calA_\phi,C^*(\wh{U}))$. We now obtain 
our `boundary module' by considering the space $\ell^2(\Z)$ with action 
of $C^*(\wh{U})$ by translations; i.e, $(\wh{U}\lambda)(m) = \lambda(m-1)$. 
We have a natural spectral triple in this setting denoted by 
$\Delta =\left(C^*(\wh{U}), \ell^2(\Z), M\right)$, where 
$M:\Dom(M)\to \ell^2(\Z)$ is given by $M\lambda(m) = m\lambda(m)$. 
It is a simple exercise to show that $\left(C^*(\wh{U}), \ell^2(\Z), M\right)$ 
is indeed a spectral triple and therefore an odd, unbounded $C^*(\wh{U})$-$\C$ 
Kasparov module. This is also what we would expect for a boundary 
system as the operator $M$ becomes the Dirac operator on the circle 
if we switch to position space. Our goal is to take the internal Kasparov 
product over $C^*(\wh{U})$ and obtain a class in $KK^0(\calA_\phi,\C)$, 
which we then link to Bellisard's spectral triple modelling a boundaryless quantum Hall system.

Whilst computing the product $\beta \hat\otimes_{C^*(\wh{U})} \Delta$ is 
relatively straight-forward, we relegate the details to the appendix and state the result. 
\begin{lemma}
\label{lem:product}
The Kasparov product of the unbounded modules 
$\beta=\left( {\calA_\phi},Z_{C^*(\wh{U})}, N\right)$ and \\ 
$\Delta =\left(C^*(\wh{U}), \ell^2(\Z), M\right)$ is given by
$$  
\beta\hat\otimes_{C^*(\wh{U})} \Delta 
=-\,\left[ \left( \calA_\phi, 
\begin{pmatrix} Z \otimes_{C^*(\wh{U})} \ell^2(\Z) \\ Z \otimes_{C^*(\wh{U})} \ell^2(\Z) \end{pmatrix}_\C, 
\begin{pmatrix} 0 & 1\hat\otimes_\nabla M -iN\hat\otimes 1 
\\ 1\hat\otimes_\nabla M +iN\hat\otimes 1 & 0 \end{pmatrix} \right)\right],  
$$
where $\calA_\phi$ acts diagonally and 
$\nabla:\calZ\to\calZ\otimes_{\mathrm{poly}(\wh{U})}\Omega^1(\mathrm{poly}(\wh{U}))$ 
is a connection on a smooth submodule $\calZ$ of $Z$ (see the Appendix). The overall minus sign
means the negative of this class in $KK(\calA_\phi,\C)$.
\end{lemma}
Our task now is to relate the product spectral triple of  Lemma \ref{lem:product} 
to the boundary-free quantum Hall system. 

\subsubsection{Equivalence of the product triple and boundary-free triple}
Recall once again~\cite{Bellissard94, CM96} our `bulk' spectral triple 
$$ 
\left( \calA_{-\phi}, \begin{pmatrix} \ell^2(\Z^2) \\ \ell^2(\Z^2) \end{pmatrix}_\C, 
\begin{pmatrix} 0 & X_1-iX_2 \\ X_1+iX_2 & 0 \end{pmatrix} \right), 
$$
where $(X_1\pm iX_2)\lambda(m,n) = (m\pm i n)\lambda(m,n)$ for 
$\lambda \in\Dom(M\pm iN)\subset\ell^2(\Z^2)$ and 
$\calA_{-\phi} \cong C^*(U,V)$ has the representation generated by
\begin{align*}
  &(U\lambda)(m,n) = e^{-2\pi i\phi n}\lambda(m-1,n),  &&(V\lambda)(m,n) = \lambda(m,n-1),
\end{align*}
with $H=U+U^*+V+V^*$ and $\lambda\in\ell^2(\Z^2)$. 
Our quantum Hall system without boundary also comes with a 
representation of $\calA_\phi\cong C^*(\wh{U},\wh{V})$ by 
magnetic translations such that the two representations commute. 
To put this another way (cf \cite{CM96}), let $\sigma(k,k')=e^{2\pi i\phi k_1'k_2}$ be 
a group 2-cocycle for $\Z^2$. Then $C^*(U,V)$ gives a right 
$\sigma$-representation of $\Z^2$ and there is a corresponding 
left $\ol{\sigma}$-representation of $\Z^2$ by $C^*(\wh{U},\wh{V})$ 
which commutes with the right representation. Because 
$C^*(U,V)\cong \calA_{-\phi} \cong \calA_{\phi}^{\mathrm{op}}$, we obtain the following.
\begin{prop}  
\label{prop:full_qH_bulk_spec_trip}
The data $\left( \calA_\phi \otimes \calA_{\phi}^\mathrm{op}, 
\ell^2(\Z^2)\oplus \ell^2(\Z^2), 
\begin{pmatrix}  0 & X_1-iX_2 \\ X_1+iX_2 & 0 \end{pmatrix},
\gamma=\begin{pmatrix} 1 & 0 \\ 0 & -1 \end{pmatrix}\right)$ defines an even spectral triple. 
\end{prop}
\begin{proof}
The only thing we need to check is that our Dirac-type operator has bounded 
commutators with a smooth subalgebra of  $C^*(\wh{U},\wh{V})$, which is an easy computation.
\end{proof}

Our aim is to reproduce this spectral triple via an explicit unitary equivalence 
with the module we have constructed via the Kasparov product. We state our central result.

\begin{thm} 
\label{thm:unitary_equiv_between_product_module_and_bulk_spec_trip}
Let $\varrho: Z_{C^*(\wh{U})}\otimes_{C^*(\wh{U})} \ell^2(\Z)\to \ell^2(\Z^2)$ 
be the map
$$  
\varrho\left(\wh{V}^{n_1-n_2}\wh{U}^m\otimes S^{n_1}(S^*)^{n_2}\otimes_{C^*(\wh{U})} e_j\right) 
= e^{-2\pi i\phi(j+m)(n_1-n_2)}e_{j+m,n_1-n_2}, 
$$
where $e_j$ and $e_{j,k}$ are the standard basis elements of 
$\ell^2(\Z)$ and $\ell^2(\Z^2)$ respectively. Then there is a representation of
$\calA_\phi\otimes \calA_\phi^\mathrm{op}$ on $Z\otimes_{C^*(\wh{U})}\ell^2(\Z)$
such that $\varrho$ gives a 
unitary equivalence between the  spectral triple 
$$ 
\left(\calA_\phi\otimes\calA_\phi^\mathrm{op},
\begin{pmatrix}Z \hat\otimes_{C^*(\wh{U})} \ell^2(\Z)\\ Z \hat\otimes_{C^*(\wh{U})} \ell^2(\Z)\end{pmatrix},
\begin{pmatrix} 0 & 1\hat\otimes_\nabla M -iN\hat\otimes 1 \\ 
1\hat\otimes_\nabla M +iN\hat\otimes 1 & 0 \end{pmatrix}\right)
$$ 
arising from the product triple of Lemma \ref{lem:product} and the bulk quantum Hall  triple in Proposition \ref{prop:full_qH_bulk_spec_trip}.
\end{thm}
\begin{proof}
We first check that, by moving elements of $C^*(\wh{U})$ across the internal tensor product, 
\begin{align*}
   \wh{V}^{n_1-n_2}\wh{U}^m\otimes S^{n_1}(S^*)^{n_2}\otimes_{C^*(\wh{U})} e_j 
   &= (\wh{V}^{n_1-n_2}\otimes S^{n_1}(S^*)^{n_2})\cdot \wh{U}^m \otimes_{C^*(\wh{U})} e_j \\
     &= \wh{V}^{n_1-n_2}\otimes S^{n_1}(S^*)^{n_2} \otimes_{C^*(\wh{U})} \wh{U}^m \cdot e_j \\
     &= \wh{V}^{n_1-n_2}\otimes S^{n_1}(S^*)^{n_2} \otimes_{C^*(\wh{U})} e_{j+m},
\end{align*}
we see that the map $\varrho$ respects the inner-products on $Z \hat\otimes_{C^*(\wh{U})} \ell^2(\Z)$ and on $\ell^2(\Z^2)$.
Hence $\varrho$ is an 
isometric isomorphism between Hilbert spaces.

Next we need to define a commuting representation of $\calA_\phi^\mathrm{op}$ 
on our product module. We can do this by pulling back the representation of 
$C^*(U,V)$ on $\ell^2(\Z^2)$ via the isomorphism $\varrho$. Alternatively, 
the same representation comes from the left action of $\calA_\phi^\mathrm{op}$ on ${}_{C^*(\wh{U})}Z$, the
module we constructed in Section \ref{subsec:left_module_over_Aop}. 
We first note that generating elements of 
$Z_{C^*(\wh{U})}\otimes_{C^*(\wh{U})} \ell^2(\Z)$ can be written as 
$\wh{V}^{n_1-n_2}\otimes S^{n_1}(S^*)^{n_2}\otimes e_j$ for some 
$j\in\Z$ and $n_1,n_2\in\N$. Then
$$  
{U}^\alpha {V}^\beta \cdot \left( \wh{V}^{n_1-n_2}\otimes S^{n_1}(S^*)^{n_2} \otimes e_j\right) 
= e^{2\pi i\phi\beta j}\wh{V}^{n_1-n_2+\beta} \otimes S^{n_1+\beta}(S^*)^{n_2} \otimes e_{j+\alpha} 
$$
for $\beta\geq 0$. A similar formula but replacing $S^{n_1+\beta}(S^*)^{n_2}$ with 
$S^{n_1}(S^*)^{n_2+|\beta|}$ gives the action for $\beta<0$. 
This left-action of $\calA_\phi^\mathrm{op}$ is compatible with the isomorphism, that is, 
$$  
\varrho\!\left[ {U}^\alpha {V}^\beta \cdot 
\left( \wh{V}^{n_1-n_2}\wh{U}^m\otimes S^{n_1}(S^*)^{n_2} \otimes e_j\right) \right] 
= {U}^\alpha {V}^\beta \cdot 
\varrho\!\left( \wh{V}^{n_1-n_2}\wh{U}^m\otimes S^{n_1}(S^*)^{n_2} \otimes e_j\right) 
$$
and this relation extends appropriately.

What remains to check is that the map $\varrho$ is compatible 
with the representation of $\calA_\phi$ and the Dirac-type operator. That is, we need to show that
\begin{align*}
   \varrho\!\left[ \wh{U}^\alpha \wh{V}^\beta \cdot 
   \left( \wh{V}^{n_1-n_2}\otimes S^{n_1}(S^*)^{n_2} \otimes e_j\right) \right] 
   &= \wh{U}^\alpha \wh{V}^\beta \cdot 
   \varrho\!\left( \wh{V}^{n_1-n_2}\otimes S^{n_1}(S^*)^{n_2} \otimes e_j\right), \\
  \varrho\!\left[(1\hat\otimes_\nabla M \pm iN\hat\otimes 1)
  \left(\wh{V}^{n_1-n_2}\otimes S^{n_1}(S^*)^{n_2} \otimes e_j\right)\right] 
  &= (X_1\pm iX_2)\varrho\!\left( \wh{V}^{n_1-n_2}\otimes S^{n_1}(S^*)^{n_2} \otimes e_j\right).
\end{align*}
For the first claim, more computations give that, for $\beta\geq 0$,
\begin{align*}
  \varrho\!\left[ \wh{U}^\alpha \wh{V}^\beta \cdot 
  \left( \wh{V}^{n_1-n_2}\otimes S^{n_1}(S^*)^{n_2} \otimes e_j\right) \right] 
  &= \varrho\!\left( e^{2\pi i\phi\alpha(\beta+n_1-n_2)} \wh{V}^{n_1-n_2+\beta}
  \otimes S^{n_1+\beta}(S^*)^{n_2}\otimes e_{j+\alpha}\right) \\
    &= e^{2\pi i\phi\alpha(\beta+n_1-n_2)}
    e^{-2\pi i\phi(j+\alpha)(\beta+n_1-n_2)}e_{j+\alpha,n_1-n_2+\beta} \\
    &= e^{-2\pi i\phi\phi j(\beta+n_1-n_2)}e_{j+\alpha,n_1-n_2+\beta}
\end{align*}
and 
\begin{align*}   
\wh{U}^\alpha \wh{V}^\beta \cdot \varrho\!\left( \wh{V}^{n_1-n_2}\otimes S^{n_1}(S^*)^{n_2} 
\otimes e_j\right) &= \wh{U}^\alpha \wh{V}^\beta e^{-2\pi i\phi j(n_1-n_2)} e_{j,n_1-n_2} \\
  &= e^{-2\pi i\phi j\beta} e^{-2\pi i\phi j(n_1-n_2)} e_{j+\alpha,n_1-n_2+\beta}.
\end{align*}
Again, the case for $\beta<0$ is basically identical. Because the result holds on generating elements, which are represented as shift operators, the result extends to the whole algebra 
and space. For the second claim, we once more check the result on 
spanning elements. We recall from the appendix that
\begin{align*}
  (1\hat\otimes_\nabla M)\left(\wh{V}^{n_1-n_2}\otimes S^{n_1}(S^*)^{n_2} \otimes e_j\right) 
  &= \wh{V}^{n_1-n_2}\otimes S^{n_1}(S^*)^{n_2} \otimes M\wh{U}^0 e_j  \\
  &= j \left(\wh{V}^{n_1-n_2}\otimes S^{n_1}(S^*)^{n_2} \otimes e_j\right).
\end{align*}
Therefore,
\begin{align*}
  \varrho\!\left[(1\hat\otimes_\nabla M \pm iN\hat\otimes 1)\left(\wh{V}^{n_1-n_2}
  \otimes S^{n_1}(S^*)^{n_2} \otimes e_j\right)\right] 
  &= (j\pm i(n_1-n_2))\varrho\!\left(\wh{V}^{n_1-n_2}\otimes S^{n_1}(S^*)^{n_2} \otimes e_j\right) \\
  &= (j\pm i(n_1-n_2))e^{-2\pi i\phi j(n_1-n_2)}e_{j,n_1-n_2} \\
  &= (X_1\pm iX_2) \varrho\!\left(\wh{V}^{n_1-n_2}\otimes S^{n_1}(S^*)^{n_2} \otimes e_m\right)
\end{align*}
and the main result follows by extending in the standard way.
\end{proof}

\begin{remark}[Factorisation and Poincar\'{e} duality]
In the proof of Theorem \ref{thm:unitary_equiv_between_product_module_and_bulk_spec_trip} the bimodule structure of $Z$ could be used to obtain the left-action of $\calA_\phi^\mathrm{op}$ on the product module.
An important 
observation is that we can either take the Kasparov product of 
$\left(\calA_\phi,Z_{C^*(\wh{U})},N\right)$ or 
$\left(\calA_\phi^\mathrm{op},{}_{C^*(\wh{U})}Z,N\right)$ 
with our boundary module \emph{and the resulting module is the same}. 
Hence we pick up an extra representation on our product module, 
which is necessary in order to completely link up the product module to the bulk spectral triple. 
The deeper meaning behind this extra structure is related to Poincar\'{e} duality for 
$\calA_\phi$: see~\cite{ConnesGrav} for more information.
\end{remark}

By setting up a unitary equivalence of spectral triples, we can 
conclude that the $K$-homological data presented in Bellissard's 
spectral triple is the same as that presented by the product module 
we have constructed. The unitary equivalence is of course much 
stronger than just stable homotopy equivalence on the level of $K$-homology.


\subsection{Pairings with $K$-Theory and the edge conductance}

We know abstractly that the $KK^1$ class defined by the Kasparov module 
$(\calA_\phi,Z_{C^*(\wh{U})},N)$ represents the boundary map in 
$K$-homology~\cite[Section 9]{Kasparov80}.  We examine this more closely by considering the pairings related to the quantisation of the Hall conductance.


We recall that the bulk spectral triple $(\calA_\phi,\ell^2(\Z^2)\oplus\ell^2(\Z^2),X,\gamma)$ 
pairs with elements in $K_0(\calA_\phi)\cong \Z[1]\oplus\Z[p_\phi]$, where 
$p_\phi$ is the Powers-Rieffel projection. For simplicity, we denote the 
corresponding $K$-homology class of our spectral triple by $[X]$, 
where we know that $[X] = [\beta]\hat\otimes_{C^*{(\wh{U})}}[\Delta]$. 
Now, $[X]$ pairs non-trivially with $[p_F]$, the Fermi projection, to 
give the Hall conductance up to a factor of $e^2/h$. Hence we have that
$$    
\sigma_H = \frac{e^2}{h}\left( [p_F] \hat\otimes_{\calA_\phi} [X] \right) 
= -\frac{e^2}{h}\left( [p_{F}] \hat\otimes_{\calA_\phi} 
\left([\beta]\hat\otimes_{C^*{(\wh{U})}}[\Delta] \right)\right),
$$
where the minus sign arises from Lemma \ref{lem:product}.
We can now use the associativity of the Kasparov product to rewrite this equation as 
$$ 
[p_{F}] \hat\otimes_{\calA_\phi} \left([\beta]\hat\otimes_{C^*{(\wh{U})}}[\Delta] \right) 
= \left([p_{F}] \hat\otimes_{\calA_\phi} [\beta]\right)\hat\otimes_{C^*{(\wh{U})}}[\Delta]. 
$$
We see that this new product 
$[p_F]\hat\otimes_{\calA_\phi} [\beta]$ is in $ KK^1(\C,C^*(\wh{U}))\cong K_1(C^*(\wh{U})) \cong \Z$,
where the last group has generator $\wh{U}$. 
So $[p_F]\hat\otimes_{\calA_\phi} [\beta]$ is represented by 
$\wh{U}^m \in C^*(\wh{U})$ for some $m\in\Z$ and we are now taking an 
odd index pairing. 

Next we note that the map
$$  
K_1(C^*(\wh{U}))\times K^1(C^*(\wh{U}))\to \Z \quad \mbox{where }
\left([p_{F}] \hat\otimes_{\calA_\phi} [\beta]\right)\times[\Delta]\mapsto 
\left([p_{F}] \hat\otimes_{\calA_\phi} [\beta]\right)\hat\otimes_{C^*{(\wh{U})}}[\Delta] 
$$
depends only on our boundary data, so this pairing is the mathematical 
formulation of the so-called edge conductance which, as we have seen, is the same as 
our bulk Hall conductance up to sign. 

Now, our definition of the edge conductance 
is purely mathematical, but one can see that the unitaries and spectral triples being used come quite naturally from considering the algebra $C^*(\wh{U})$ acting on $\ell^2(\Z)$, 
which is exactly what we would consider as a `boundary system' in 
the discrete picture. Hence our approach to the edge conductance is 
physically reasonable. Furthermore, the computation of the edge 
conductance boils down to computing 
$\mathrm{Index}\!\left(\Pi\wh{U}^m\Pi\right)=-m$ for $\Pi:\ell^2(\Z)\to\ell^2(\N)$, 
which is a much easier computation than $ [p_F]\hat\otimes_{\calA_\phi}[X]$. 


\section*{Appendix: Computing the odd Kasparov product}

It is proved in \cite[Theorem 7.5]{KL13} that the $KK$-class of the product
$$
\left[\left( {\calA_\phi},Z_{C^*(\wh{U})}, N\right)\right]\otimes_{C^*(\wh{U})}
\left[\left(C^*(\wh{U}), \ell^2(\Z), M\right)\right]
$$
is represented by
$$
\left( \calA_\phi, 
\begin{pmatrix} Z \otimes_{C^*(\wh{U})} \ell^2(\Z) \\ Z \otimes_{C^*(\wh{U})} \ell^2(\Z) \end{pmatrix}_\C, 
\begin{pmatrix} 0 & N\hat\otimes 1 -i 1\hat\otimes_\nabla M
\\ N\hat\otimes 1+i1\hat\otimes_\nabla M & 0 \end{pmatrix} \right).
$$ 
There are several conditions to check in order to apply \cite[Theorem 7.5]{KL13}, but the 
product we are taking turns out to be of the simplest kind, and we omit these simple checks.
Here $\calA_\phi$ acts diagonally on column vectors, and the grading is 
$\begin{pmatrix} 1 & 0\\ 0 & -1\end{pmatrix}$. To define $1\otimes_\nabla M$,
we let 
$\calZ_{C^*(\wh{U})}$ be the submodule of $Z$ given by finite 
sums of elements $\wh{V}^{n_1-n_2}\wh{U}^m\otimes S^{n_1}(S^*)^{n_2}$ and take the connection
$$  
\nabla :\calZ \to \calZ \otimes_{\mathrm{poly}(\wh{U})} \Omega^1(\mathrm{poly}(\wh{U})) 
$$
given by
$$  
\nabla\left(\sum_{n_1,n_2,m} z_{n_1,n_2,m}\right) 
= \sum_{n_1,n_2}z_{n_1,n_2,0}\otimes \delta(\wh{U}^m), 
$$
where $\delta$ is the universal derivation, and we represent $1$-forms on $\ell^2(\Z)$ via
$$  
\tilde{\pi}\left(a_0\delta(a_1)\right)\lambda = a_0[M,a_1]\lambda  
$$
for $\lambda\in\ell^2(\Z)$. We define
$$  
(1\otimes_\nabla M)(z\otimes\lambda) 
:= (z\otimes M\lambda) + (1\otimes \tilde{\pi})\circ(\nabla\otimes 1)(x\otimes \lambda). 
$$
The need to use a connection to correct the naive formula $1\otimes M$ is because 
$1\otimes M$ is not well-defined on the balanced tensor product. Computing yields that
\begin{align*}
   (1\otimes_\nabla M)\left(\sum_{n_1,n_2,\beta}z_{n_1,n_2,\beta}\otimes \lambda\right) 
   &= \sum_{n_1,n_2}z_{n_1,n_2,0}\otimes \wh{U}^\beta M \lambda 
   + \sum_{n_1,n_2}z_{n_1,n_2}\otimes [M,\wh{U}^\beta]\lambda \\
   &= \sum_{n_1,n_2}z_n \otimes M\wh{U}^\beta \lambda.
\end{align*}

Now conjugating the representation, operator and grading by  
$\begin{pmatrix} 0 & i\\ 1 & 0\end{pmatrix}$ yields the unitarily equivalent spectral triple
$$
\left( \calA_\phi, 
\begin{pmatrix} Z \otimes_{C^*(\wh{U})} \ell^2(\Z) \\ Z \otimes_{C^*(\wh{U})} \ell^2(\Z) \end{pmatrix}_\C, 
\begin{pmatrix} 0 &  -( 1\hat\otimes_\nabla M-iN\hat\otimes 1)
\\-(1\hat\otimes_\nabla M+i N\hat\otimes 1) & 0 \end{pmatrix} \right)
$$
with grading $\begin{pmatrix} -1 & 0\\ 0 & 1\end{pmatrix}$. In turn, the $KK$-class of this
spectral triple is
given by
$$
-\,\left[\left( \calA_\phi, 
\begin{pmatrix} Z \otimes_{C^*(\wh{U})} \ell^2(\Z) \\ Z \otimes_{C^*(\wh{U})} \ell^2(\Z) \end{pmatrix}_\C, 
\begin{pmatrix} 0 &   1\hat\otimes_\nabla M-iN\hat\otimes 1
\\1\hat\otimes_\nabla M+i N\hat\otimes 1 & 0 \end{pmatrix} \right)\right]
$$
with grading $\begin{pmatrix} 1 & 0\\ 0 & -1\end{pmatrix}$.

\end{document}